\documentclass[]{article}
\usepackage{amsmath}
\usepackage{amssymb}
\usepackage[noend]{algorithmic}
\usepackage{algorithm}
\usepackage{fullpage}
\usepackage{graphicx}

\newcommand{\SVD}{\operatorname{SVD}}
\newcommand{\R}{{\mathbb R}}
\renewcommand{\S}{{\mathbb S}}
\newcommand{\E}{{\mathbb E}}
\newcommand{\N}{{\mathcal N}}
\newcommand{\eps}{\varepsilon}
\newcommand{\diag}{\operatorname{diag}}
\newcommand{\FI}{Frequent-Items}
\newcommand{\FD}{Frequent-Directions}
\newcommand{\myfigurewidth}{0.7}
\usepackage{amsthm}

\newtheorem{claim}{Claim}

\newtheorem{remark}{Remark}

\title{Simple and Deterministic Matrix Sketching}
\author{Edo Liberty\thanks{Yahoo! Research}}
\date{\nonumber}

\begin{document}
\maketitle

\begin{abstract}
We adapt a well known streaming algorithm for approximating item frequencies to the matrix sketching setting. 
The algorithm receives the rows of a large matrix $A \in \R^{n \times m}$ one after the other in a streaming fashion. 
It maintains a sketch matrix $B \in \R^ { 1/\eps  \times m}$ such that for any unit vector $x$
\[
\|Ax\|^2 \ge \|Bx\|^2 \ge  \|Ax\|^2 - \eps \|A\|_{f}^2 \ .
\]
Sketch updates per row in $A$ require $O(m/\eps^2)$ operations in the worst case.
A slight modification of the algorithm allows for an amortized update time of $O(m/\eps)$ operations per row. 
The presented algorithm stands out in that it is: deterministic, simple to implement, and elementary to prove.
It also experimentally produces more accurate sketches than widely used approaches while still being computationally competitive.
\end{abstract}

\section{Introduction}
Efficiently obtaining compact approximations, or sketches, of large matrices is a very common and important problem.
It is essential for obtaining approximate matrix Singular Value Decompositions (SVD) or low rank approximations of large matrices.
It is used in large scale data mining, Latent Semantic Indexing (LSI), $k$-means clustering, Principal Component Analysis (PCA), Linear regression, solving large linear systems, matrix preconditioning, and many other important and commonly performed tasks.

Moreover, modern large data sets are often viewed as large matrices.
For example, textual data in the bag-of-words model is stored such that each row in the data matrix 
corresponds to one document and non zero entries correspond to words in the documents.
Such matrices are often extremely large and distributed across many machines.
Sketching methods, therefore, are designed to be pass efficient which means the data is read at most a constant number of times.
If only one pass is required the computational model is also referred to as the streaming model. 
The streaming model is especially attractive since the analysis can be performed immediately when the data is collected.
In that case its storage can be eliminated altogether. 

There are three main matrix sketching approaches which are presented here in an arbitrary order.
The first generates a sparser version of the matrix. 
Sparser matrices are stored more efficiently and can be multiplied by other matrices 
faster \cite{AroraHazanKale2006}\cite{AchlioptasMcsherry2007}.
The second approach is to randomly combine matrix rows  \cite{PapadimitriouTRV1998}\cite{Vempala2004}\cite{Sarlos06}\cite{tygert07PNAS}.
These methods rely on properties of random low dimensional subspaces and strong concentration of measure phenomena.
The third is to find a small subset of matrix rows (or columns) which approximate the entire matrix. 
This problem is known as the Column Subset Selection Problem and has been thoroughly investigated 
\cite{FriezeKannanVempala1998}\cite{Drineas03passefficient}\cite{BoutsidisMahoneyDrineas2009}\cite{DeshpandeV06}\cite{DrineasMohoneyMuthukrishnan2011}\cite{BoutsidisDrineasMagdon2011}. Recent results obtain algorithms with almost matching lower bounds \cite{DeshpandeV06}\cite{BoutsidisDrineasMagdon2011}\cite{ClarksonWoodruff2009}.
Alas, It is not immediately clear how to compare these methods' results to ours since their objectives are different.
They aim to recover a low rank matrix whose column space contains most of space spanned by the matrix top $k$ singular vectors.
Moreover, most of the above algorithms require several passes over the input matrix.
A simple streaming solution to the Column Subset Selection problem is obtained by sampling columns (rows in this case) from the input matrix.
The rows are sampled with probability proportional to their squared norm.
Despite this algorithm's apparent simplicity, providing tight bounds for its performance required over a decade of research \cite{FriezeKannanVempala1998}\cite{AhlswedeW02}\cite{DrineasKannan2003}\cite{RudelsonVershyninMatrixSampling2007}\cite{VershyninMatrixChernoffBounds}\cite{Oliviera2010}\cite{DrineasMohoneyMuthukrishnan2011}. 

This manuscript proposes a forth approach which draws on the matrix sketching problem's similarity to the item frequency estimation problem.
In what follows, we shortly describe the item frequency approximation problem and a well known solution for it
which our proposed algorithm will be based on.

In the item frequency approximation problem there is a universe 
of $m$ items $a_1,\ldots,a_m$ and a stream $A_1,\dots,A_{n}$ of item appearances.
The frequency of an item is the number of times it appears in the stream.
It is trivial to produce these counts using $O(m)$ space simply by keeping a counter for each item.
Our goal is to use $o(m)$ space and produce approximate 
frequencies $g_j$ such that $ |f_j - g_j | \le \eps n$ for all $j$ and prescribed precision $\eps$.
%
%

This problem received an incredibly simple and beautiful solution by \cite{Misra1982}.
This solution was later and independently rediscovered 
by \cite{DemaineLopezAlejandroMunro2002} and \cite{Karp03} which also improved its update complexity.
The algorithm simulates the process of repeatedly `deleting' form the 
stream $\ell = 1/\eps$ appearances of {\it different} items until this is no longer possible. 
In other words, until there are less than $\ell$ unique items left. 
This trimmed stream is stored concisely in $O(\ell)$ space.
The claim is that if item $a_j$ appears in the final trimmed stream $g_j$ times 
than $g_j$ is a good approximation for its true frequency $f_j$ (even if $g_j=0$).
This is because $f_j - g_j \le t$ where $t$ is the number of times items were deleted (item of type $a_j$ is deleted at most once in each deletion batch). 
Moreover, we delete $\ell$ items in every batch and at most $n$ items can be deleted altogether. Thus, $t\ell \le n$ or $t \le \eps n$ which completes the proof.
The reader is referred to \cite{Karp03} for an efficient streaming implementation.
From this point on we refer to this algorithm as \FI. 

Let us now describe the item frequency problem as a matrix sketching problem. 
Let $A$ be a stream of indicator vectors $A_{i} \in \R^{m}$ instead of discrete elements.
Each row in $A$ is $A_i = e_j$ (the $j$'th standard basis vector) if $A_i = a_j$.
The frequency $f_j$ can be expressed as $f_j = \|A e_j\|^{2}$. 
A good approximation matrix $B$ would be one such that $g_j = \|B e_j\|^{2}$ is a good approximation to $f_j$.
Replacing $n= \|A\|_{f}^{2}$ we get that the condition $| f_j - g_j | \le \eps n$
is equivalent to $|\|Ae_j\|^2 - \|Be_j\|^2| \le \eps \|A\|_{f}^{2}$. 
From the above, it is clear that for `item indicator' matrices a sketch $B \in \R^{1/\eps \times m}$ can be
obtained by the \FI~algorithm.

In this paper we suggest \FD~which is an extension of \FI~to general matrices. 
Given any matrix $A \in \R^{n \times m}$ a sketch $B \in \R^{1/\eps \times m}$ is produced such that:
\[
\forall x \in \S^{m-1} \;\; \left|\|Ax\|^2 - \|Bx\|^2\right|  \le \eps \|A\|_{f}^{2} \mbox{\;\;\;\;\; or alternatively \;\; \;\;\;}  \|A^{T}A -  B^{T}B\| \le \eps\cdot tr(A^{T}A)
\]
The intuition behind \FD~is surpassingly similar to the one above.
In the same way that \FI~periodically deletes $\ell$ different element, \FD~periodically removes from its sketch $\ell$ orthogonal vectors.
This means that the Frobenius norm of the trimmed sketch matrix reduces by a factor $\ell$ faster than its projection on any single direction.
Since the final sketch's Frobenius norm is non negative, we are guarantied that no direction in space is reduced by `too much'. 
This intuition exact below.
As a remark, when presented with and `item indicator' matrix \FD~exactly mimics \FI.

\section{The algorithm}\label{thealg}

\begin{algorithm} 
\caption{\FD}
\label{alg}
\begin{algorithmic}
\STATE {\bf Input:} $\eps \in (0,1], \;\;A \in \R^{n \times m}$ 
\STATE $\ell \leftarrow \lceil 1/\eps \rceil$
\STATE $B \leftarrow $ all zeros matrix $\in \R^{\ell \times m}$  
\FOR{$i \in [n]$}
	\STATE $B_{\ell} \leftarrow A_i$  \hfill \# $B_{\ell}$ denotes the $\ell$'th row of $B$
	\STATE $[U,\Sigma,V] \leftarrow SVD(B)$  \hfill \# $U \Sigma V = B$,  \;$U^{T}U = VV^{T} = I_{\ell}$
	\STATE \hfill \# $\Sigma = \diag([\sigma_1,\ldots,\sigma_\ell])$,  \;$\sigma_1 \ge \ldots \ge \sigma_\ell$
	\STATE $C \leftarrow  \Sigma V$  \hfill \# Only for proof notation
	\STATE $\delta_i \leftarrow \Sigma_{\ell,\ell}^2$ \hfill \# $\Sigma_{\ell,\ell}$ the least singular value of $B$
	\STATE $\bar{\Sigma} \leftarrow \sqrt{\max(\Sigma^2 - I_{\ell}\delta_i,0)}$ \hfill \# $I_{\ell}$ is an $\ell \times \ell$ identity matrix
	\STATE $B \leftarrow \bar{\Sigma} V$ \hfill \# Here $B_{\ell}$ contains zeros since $\bar{\Sigma}_{\ell,\ell} = 0$
\ENDFOR
\STATE {\bf Return:} $B$ 
\end{algorithmic}
\end{algorithm}

\begin{claim}
Let $B$ be the result of applying Algorithm~\ref{alg} to matrix $A$ with prescribed precision parameter $\eps$ then:
\[
\forall x\in \R^{d} \;\;\;\;  \|Ax\|^2  \ge \|Bx\|^2 \ge \|Ax\|^2 - \eps \|A\|_{f}^{2}\|x\|^2
\] 
\noindent Or alternatively:
\[
A^{T}A \succeq B^{T}B \mbox{\;\; and \;\;} \|A^{T}A - B^{T}B\| \le \eps\cdot tr(A^{T}A) 
\]
\end{claim}

\begin{proof}
We begin by obtaining the value of $\sum_{i=1}^{n}\delta_i$ by computing the Frobenius norm of $B$.
Let $B^i$, $C^{i}$ and $V^{i}$ be the values of $B$, $C$ and $V$ after the main loop in the algorithm is executed $i$ times.
For example, $B^0$ is an all zeros matrix and $B^n$ is the returned sketch matrix.
\begin{eqnarray*}
 \|B^n\|_{f}^{2} &=& \sum_{i=1}^{n} [\|B^{i}\|_{f}^{2} - \|B^{i-1}\|_{f}^{2}] \\
&=& \sum_{i=1}^{n} [(\|C^{i}\|_{f}^{2} - \|B^{i-1}\|_{f}^{2}) - (\|C^{i}\|_{f}^{2} - \|B^{i}\|_{f}^{2})]\\
&=& \sum_{i=1}^{n} \|A_i\|^2 - tr({C^{i}}^{T}C^{i}- {B^{i}}^{T}B^{i})  \\
&=&  \|A\|_{f}^2 - \sum_{i=1}^{n} tr(\delta_i {V^{i}}^{T}V^{i}) = \|A\|_{f}^{2} - \ell\sum_{i=1}^{n}\delta_i
\end{eqnarray*}
The reason that $\|C^{i}\|_{f}^{2} - \|B^{i-1}\|_{f}^{2} = \|A_i\|^2$ is because $C^{i}$ is, up to a unitary left rotation, a matrix which contains both $B^{i-1}$ and $A_i$. Remember that the last row of $B^{i-1}$, which $A_i$ replaced, contains only zero values.
We gain that $\sum_{i=1}^{n}\delta_i = (\|A\|_{f}^{2}-\|B\|_{f}^{2})/\ell$ which we use shortly.
Now, let us compute the value of $\|Ax\|^2 - \|Bx\|^2$ for a vector $x \in \R^d$. 
\begin{eqnarray*}
\|Ax\|^2 - \|Bx\|^2 &=& \sum_{i=1}^{n}[ \langle A_i,x\rangle^2 + \|B^{i-1}x\|^2 - \|B^{i}x\|^2] \\
&=&  \sum_{i=1}^{n}[\|C^{i}x\|^2 - \|B^{i}x\|^2] \\
&\le& \sum_{i=1}^{n}\|{C^{i}}^{T}C^{i} - {B^{i}}^{T}B^{i}\|\cdot \|x\|^{2}  = \|x\|^{2} \sum_{i=1}^{n}\delta_i
\end{eqnarray*}
The first transition is due to the fact that $\langle A_i,x\rangle^2 + \|B^{i-1}x\|^2 = \|C^{i}x\|^2$.
The second transition is correct because $\|{C^{i}}^{T}C^{i} - {B^{i}}^{T}B^{i}\|= \|\delta_i {V^{i}}^{T}V^{i}\| = \delta_i$.
Substituting that $\sum_{i = 1}^{n} \delta_i = (\|A\|_{f}^{2}-\|B\|_{f}^{2})/\ell$, $\|B\|_{f}^{2} \ge 0$ and $1/\ell \le \eps$ completes the claim.
\end{proof}

\subsection{Running time}\label{run}
Let $T_{\SVD}(\ell,m)$ stand for the number of operations required to obtain the Singular Value Decomposition of an $\ell$ by $m$ matrix.
The worst case update time of \FD is $O(T_{\SVD}(\ell,m))$ which is also $O(m/\eps^2)$ operations per incoming vector. 
This is because the execution of the main loop is dominated by computating  the sketch $\SVD$.

However, there is no reason to compute the $\SVD$ of $B$ in each and every iteration.
In Algorithm \ref{alg}, consider replacing the statement 
$\delta_i \leftarrow \Sigma_{\ell,\ell}^2$ with $\delta_i \leftarrow \Sigma_{c\ell,c\ell}^2$
for some $c \in [1/10,9/10]$ and assume for simplicity that $c\ell$ is integer.
In every computation of the $\SVD$ the algorithm nullifies $(1-c)\ell$ rows of the sketch.
Therefore, the next $(1-c)\ell$ rows it receives can be places in zero valued rows.
Consequently, the $\SVD$ of the sketch is computed only once every $(1-c)\ell$ input rows.
This gives a total running time of $O(n/\ell \cdot T_{\SVD}(\ell,m))$ which is amortized $O(m/\eps)$ per row.
The proof above carries over almost without a change. The resulting approximation is slightly weakened though.
The modified algorithm only guaranties that $\|A^{T}A - B^{T}B\| \le tr(A^{T}A)/c\ell$.

\begin{remark}
From a theoretical stand point $O(T_{\SVD}(\ell,m))$ can be reduced using fast matrix multiplication.
Let $\alpha$ denote the smallest scalar such that multiplying two $\ell \times \ell$ matrices requires $O(\ell^{2+\alpha})$ operations.
To compute the $\SVD$ of $B$ we first use fast matrix multiplication to compute $BB^T$ in time $O(m\ell^{1+\alpha})$.
Then we compute $[U,S^2,U^{T}] = \SVD(BB^T)$ which requires $O(\ell^3)$ operations.
Finally, we compute the right singular vectors of $B$ by $V = S^{-1}U^{T}B$ which, using fast matrix multiplication, requires $O(m\ell^{1+\alpha})$. This reduces the theoretical computation time of the $\SVD$ to $O(m \ell^{1+\alpha} + \ell^3)$. 
Note that $\alpha < 0.5$ \cite{CohnKSU05}.
Although computing the $\SVD$ in this manner is numerically unstable and generally recommended against, 
in this case it might be beneficial.
Due to the algorithm's relatively weak approximation guaranty the accuracy loss incurred by squaring the matrix condition number might not be meaningful.
Moreover, there is no need for the $\SVD$ procedure to converge to machine precision. 

\end{remark}

\subsection{Parallelization and sketching sketches}
A convenient property of this sketching technique is that it allows for combining sketches.
In other words, let $A_1$ and $A_2$ denote two halves of a larger matrix $A$. 
Also, let $B_1$ and $B_2$ be the sketches computed by the above technique for $A_1$ and $A_2$ respectively.
Now let the final sketch, $C$, be the sketch of a matrix  $B$ which contains both $B_1$ and $B_2$.
It still holds that $\|A^{T}A - C^{T}C\| \le \eps\cdot tr(A^{T}A - C^{T}C)$.
To see this we compute $\|Cx\|^2$ for a test vector $\|x\| = 1$.
\begin{eqnarray*}
\|Cx\|^2 &\ge& \|Bx\|^{2} - \eps (\|B\|_{f}^{2} - \|C\|_{f}^{2}) \\
&=& \|B_{1}x\|^2 + \|B_{2}x\|^2  - \eps (\|B_1\|_{f}^{2} +\|B_2\|_{f}^{2}) +  \eps \|C\|_{f}^{2}\\
&\ge& \;\,\,\, \|A_{1}x\|^2 - \eps (\|A_{1}\|_{f}^2 - \|B_{1}\|_{f}^{2}) \\
&& +\|A_{2}x\|^2 - \eps  (\|A_{2}\|_{f}^2 - \|B_{2}\|_{f}^{2}) \\ 
&& - \eps (\|B_1\|_{f}^{2} +\|B_2\|_{f}^{2}) + \eps \|C\|_{f}^{2} \\
&=& \|A_{1}x\|^2 + \|A_{2}x\|^2 - \eps (\|A_{1}\|_{f}^2 + \|A_{2}\|_{f}^2) + \eps \|C\|_{f}^{2}\\
&=& \|Ax\|^{2} - \eps  (\|A\|_{f}^2 -\|C\|_{f}^{2})
\end{eqnarray*}
Here we use the fact that $\|B_{1}x\|^{2} \ge \|A_{1}x\|^{2}  - \eps (\|A_{1}\|_{f}^{2}- \|B_{1}\|_{f}^2) $ for $\|x\|=1$ which is a consequence of the derivation above.
This property is especially useful when the matrix (or data) is distributed across many machines which is often the case in modern large scale data.

\subsection{Connection to matrix low rank approximation}
Low rank approximation of matrices is a well studied problem.
The goal is to obtain a small matrix $B$ containing $\ell$ columns 
which contains in its columns space a matrix $\Pi$ of rank $k$ such that $\|A - A \Pi\|_{\xi} \le (1+\eps)\|A - A_k\|_\xi$.
Here, $A_k$ is the best rank $k$ approximation of $A$ and $\xi$ is either $2$ (spectral norm) or $f$ (Frobenius norm). 
It is difficult to compare our algorithm to this line of work since the types of bounds sought are qualitatively different.
We remark, however, that it is possible to use \FD~to produce a low rank approximation result.\\
\noindent {\bf Lamma~$4$ from \cite{DrineasKannan2003} (modified)}.
Let $P^{B}_{k}$ denote the projection matrix on the left $k$ singular vectors of $B$ corresponding to its largest singular values.
Then the following holds $\|A - AP^{B}_{k}\|^2 \le \sigma^{2}_{k+1} + 2\|A^{T}A - B^{T}B\|$ where $\sigma_{k+1}$ is the $(k+1)$'th
singular value of $A$.

Therefore, if  $2\|A^{T}A - B^{T}B\| \le \eps \sigma^{2}_{k+1}$ we have that $\|A - AP^{B}_{k}\| \le \sigma_{k+1}(1+\eps)$
which is a $1+\eps$ approximation to the optimal solution.
Letting the sketch $B$ maintain $\ell \ge 2\|A\|^{2}_{f}/\eps \sigma^{2}_{k+1}$ ensures that this is the case.
Since $\|A\|_{f}^{2} / \sigma_{k+1}^{2} \in \Omega(k)$ this is asymptotically inferior to the space requirement of \cite{BoutsidisDrineasMagdon2011}. That said, if $\|A\|_{f}^{2}/\sigma_{k+1}^{2} \in O(k)$ this is also optimal due to \cite{ClarksonWoodruff2009}.

\section{Experiments}\label{experiments}

We compere \FD~to five different techniques. 
The first two constitute a brute force and a na\"ive baseline. 
The other three are common algorithms which are commonly used in practice.
Namely: sampling, hashing, and random projection. 
These produce sketch matrices $B \in \R^{\ell \times m}$ such that $\E[B^{T} B] = A^T A$. 
The tested methods are limited in storage to an $\ell \times m$ sketch matrix $B$ and additional auxiliary variables in $o(\ell  m)$ space.
This is with the exception of the brute force algorithm.
For a given input matrix $A$ we compare the methods' computational efficiency and resulting sketch accuracy.
The computational efficiency is taken as the time required to produce $B$ from the stream of $A$'s rows. 
The accuracy of a sketch matrix $B$ is measured by $\|A^{T}A - B^{T}B\|$.

\noindent{\bf Brute Force:} the brute force approach produces the optimal rank $\ell$ approximation of $A$.
It explicitly computes the matrix $A^{T}A = \sum_{i}^{n} A_{i}^{T} A_{i}$ by aggregating the outer products of the rows of $A$. 
The final `sketch' consists of the top $\ell$ right singular vectors and 
values (square rooted) of $A^{T}A$ which are obtained by computing its $\SVD$.
The update time per row in $A$ is $O(m^2)$ and space requirement is $\Theta(m^2)$.

\noindent{\bf Na\"ive:} upon receiving a row in $A$ the na\"ive method does nothing. The sketch it returns is an all zeros $\ell$ by $m$ matrix.
This baseline is important for two reasons. 
First, it can actually be more accurate than random methods due to under sampling scaling issues.
Second, although it does not perform any computation is does incur computation overheads and I/O exactly like the other methods.
It is therefore an important benchmark in both accuracy and running time measurements.

\noindent {\bf Sampling:} each row in the sketch matrix $B^{samp}$ is chosen i.i.d. from $A_i$ and rescaled.
More accurately, each row $B^{samp}_j$ takes the value $\frac{1}{\sqrt{\ell}}\frac{\|A\|_f}{\|A_i\|}A_i$ with probability $p_i = \|A_i\|^2/\|A\|_f^2$.
The space it requires is $O(m \ell)$ in the worst case but it can be much lower if the chosen rows are sparse.
Here this is implemented as $\ell$ independent reservoir samplers, each sampling one row according to the distribution.
The update running time is therefore, $O(\ell + m)$ per row in $A$.

\noindent {\bf Hashing:} The matrix $B^{hash}$ is generated by adding or subtracting the rows of $A$ from random rows of $B^{hash}$.
More accurately, $B^{hash}$ is initialized to be an $\ell$ by $m$ all zeros matrix. 
Then, when processing $A_i$ we perform $B^{hash}_{h(i)} \leftarrow B^{hash}_{h(i)} + s(i)A_i$.
Here $h:[n]\rightarrow [\ell]$ and $s:[n]\rightarrow \{-1,1\}$ are perfect hash functions. 
There is no harm in assuming such functions exist since complete randomness is na\"ively possible without dominating either space or running time.
This method is often used in practice by the machine learning community and is referred to as `feature hashing' or `hashing trick'  \cite{WeinbergerDLSA2009}.

\noindent {\bf Random Projection:} 
The matrix $B^{proj}$ is equivalent to the matrix $RA$ where $R$ is an $\ell$ by $d$ matrix such that $R_{i,j} \in \{-1/\sqrt{\ell},1/\sqrt{\ell}\}$ uniformly. 
Since $R$ is a random projection matrix \cite{Achlioptas01} $B^{proj}$ contains the $m$ columns of $A$ randomly projected from dimension $n$ to dimension $\ell$. 
This is easily computed in a streaming fashion while requiring at most $O(m \ell)$ space and $O(m \ell)$ operation per row updated.
For proofs of correctness and usage see \cite{PapadimitriouTRV1998}\cite{Vempala2004}\cite{Sarlos06}\cite{tygert07PNAS}.

\noindent {\bf \FD:} This indicates the modified algorithm described in Section~\ref{run} with $c=1/3$. 
So, while it requires a sketch matrix of size $\ell \times m$ it night actually return a sketch of rank $\ell/3$.
Moreover, it only guaranties that $\|A^{T}A - B^{T}B\| \le 3\cdot tr(A^{T}A)/\ell$.
The benefit, however, is that its amortized running time is $O(m\ell)$ per row.


The generated input matrices $A$ contains $d$ dimensional signal and $m$ dimension noise. 
More accurately $A = S D U + N/\zeta$.
The signal coefficients matrix $S \in  \R^{n \times d}$ is such that $S_{i,j} \sim \N(0,1)$ i.i.d.
The diagonal matrix $D$ is $D_{i,i} = 1- (i-1)/d$ which gives linearly diminishing signal singular values. 
The signal row space matrix $U \in \R^{d \times m}$ contains a random $d$ dimensional subspace in $\R^{m}$, for clarity, $UU^{T} = I_{d}$.
The matrix $S D U$ is exactly rank $d$ and constitutes the signal we wish to recover. 
The matrix $N \in \R^{n \times m}$ contributes additive Gaussian noise $N_{i,j} \sim \N(0,1)$. 
Due to \cite{Vershynin08prodOfRandMatrices}, the spectral norms of $S D U$ and $N$ are expected to be the same up to some constant.
Experimentally, this constant is close to $1$.
Therefore, when the signal to noise ratio $\zeta$ is close to (or less than) $1$ we cannot expect to approximate $A^{T}A$ since the noise dominates the signal.
On the other hand, when $\zeta \in \omega(1)$ the spectral norm is dominated by the signal which is therefore recoverable.
As a remark, note that the Frobenius norm of $A$ is dominated by the noise for $\zeta \in o(\sqrt{m/d})$.

The values used in the experiments are $n = 10,000$, $m=1000$, $\ell=[10,20,\ldots,300]$, $d = [5,10,20,50,100]$, $\zeta = [1,2,\ldots,15]$.
Each method produced a sketch for each matrix, $A$, which is generated according to every parameter combination.
Each resulting sketch $B$ was measured for accuracy which is defined as $\|A^{T}A - B^{T}B\|$.
The running time for producing each sketch by the different methods was also measured.
The entire experiment was repeated $7$ times and the reported results are median values of these independent executions.
The experiments were conducted on a FreeBSD machine with 50GB RAM, and 12MB cache using a single Intel(R) Xeon(R) X5650 CPU.  
Example results are plotted and explained in Figures~\ref{diff_two_vs_ell}, \ref{running_time} and  \ref{diff_two_vs_snr}.

\begin{figure}[htbp]
\begin{center}
\includegraphics[width=\myfigurewidth\textwidth]{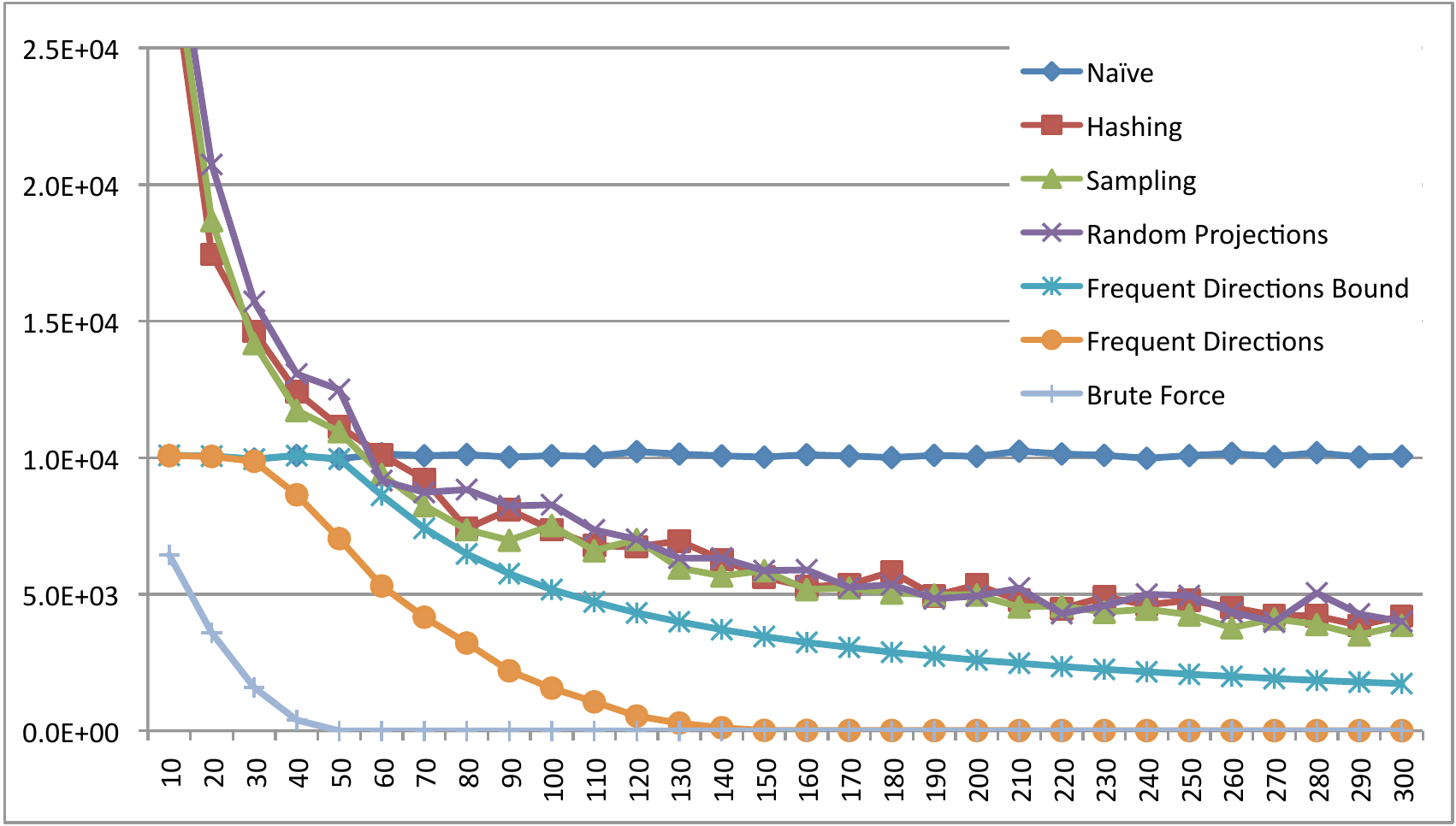}
\caption{Accuracy vs. sketch size. The $y$-axis indicates the accuracy of the sketches.
If a method returns a sketch matrix $B$ the accuracy is measured by $\|A^{T}A - B^{T}B\|$.
The size of the sketch is fixed for all methods and is $B \in \R^{\ell \times m}$.
The value of $\ell$ is indicated on the $x$-axis.
The form of the input matrix is explained in Section~\ref{experiments}. 
Here the signal dimension is $d=50$ and the signal to noise ratio is $\zeta = 10$.
There are a few interesting observations here. First, all three random techniques actually perform worse than
na\"ive for small sketch sizes. This is not the case with \FD. 
Second, the three random techniques perform equally well. This might be a result of the chosen input.
Nevertheless, practitioners should consider these methods as comparable alternatives.
The sketching technique performs significantly better than all three.
The curve indicated by ``\FD~Bound" plots the accuracy guarantied by \FD. 
While the worst case bound for \FD~is consistently better than the three competing techniques
it is still far from being tight for \FD~for large sketch sizes. 
Notice that \FD~reaches its best result already at $\ell = 150$.
This is because the signal dimension is $d=50$ and \FD is implemented with parameter $c=1/3$ (see Section~\ref{run}).
}
\label{diff_two_vs_ell}
\end{center}
\end{figure}

\begin{figure}[htbp]
\begin{center}
\includegraphics[width=\myfigurewidth\textwidth]{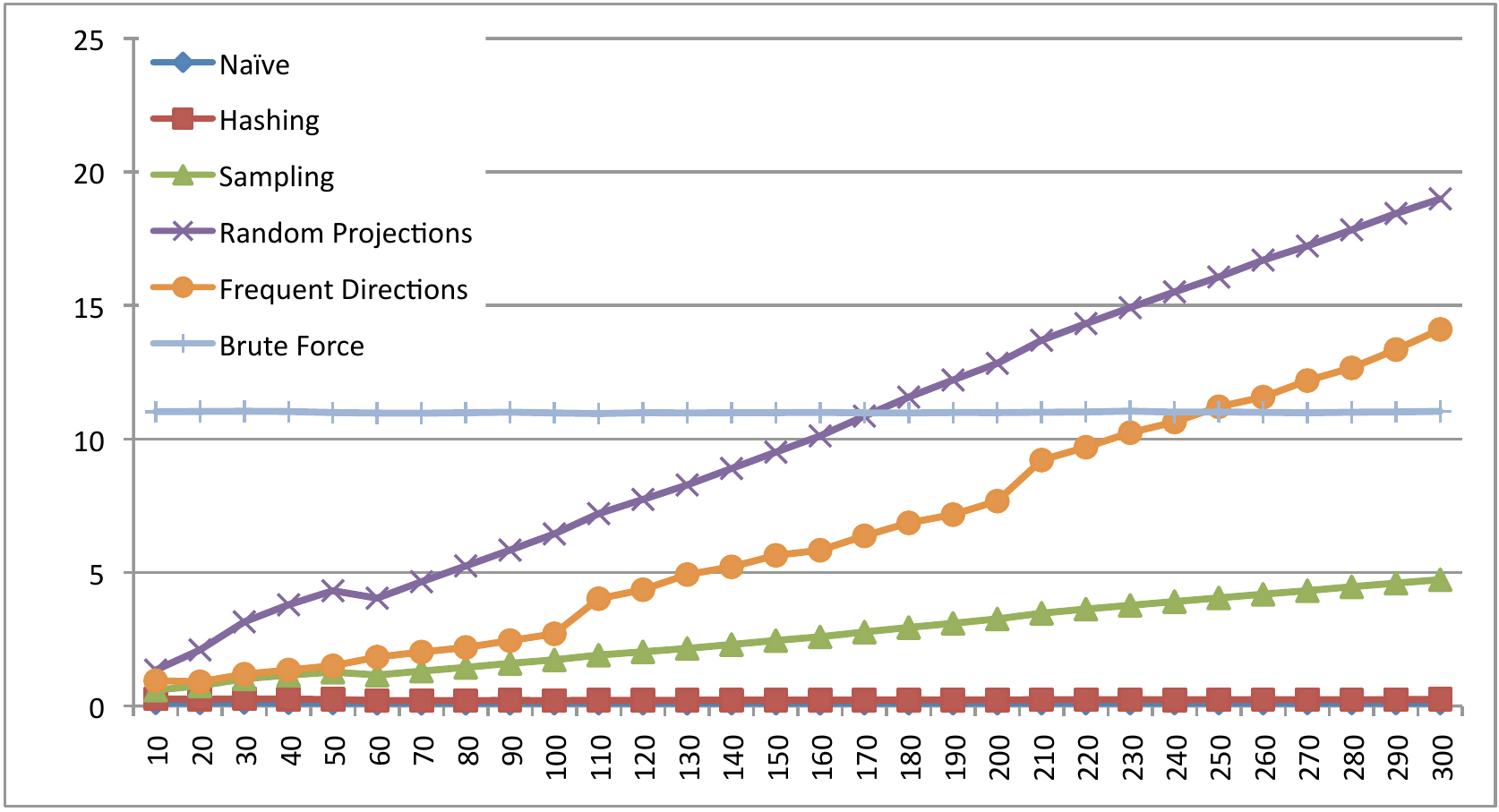}
\caption{Running time in seconds vs. sketch size. 
Each method produces a sketch matrix $B$ of size $\ell \times 1000$ for a dense $10,000 \times 1,000$ matrix. 
The value of $\ell$ is indicated by the $x$-axis. 
The total amount of computation time required to produce the sketch is indicated on the $y$-axis in seconds.
The brute force method computes the complete $\SVD$ of $A$ and therefore its running time is independent of $\ell$.
Note that Hashing is almost as fast as the Na\"ive method and independent of $\ell$ which is expected.
The rest of the methods exhibit a linear dependence on $\ell$ which is also expected.
Surprisingly though, sketching is more computationally efficient than random projections although both require $O(m \ell)$ operations per input row.
It is important to stress that the implementations above are not very well optimized. 
Different implementations might lead to slightly different results.
}
\label{running_time}
\end{center}
\end{figure}

\begin{figure}[htbp]
\begin{center}
\includegraphics[width=\myfigurewidth\textwidth]{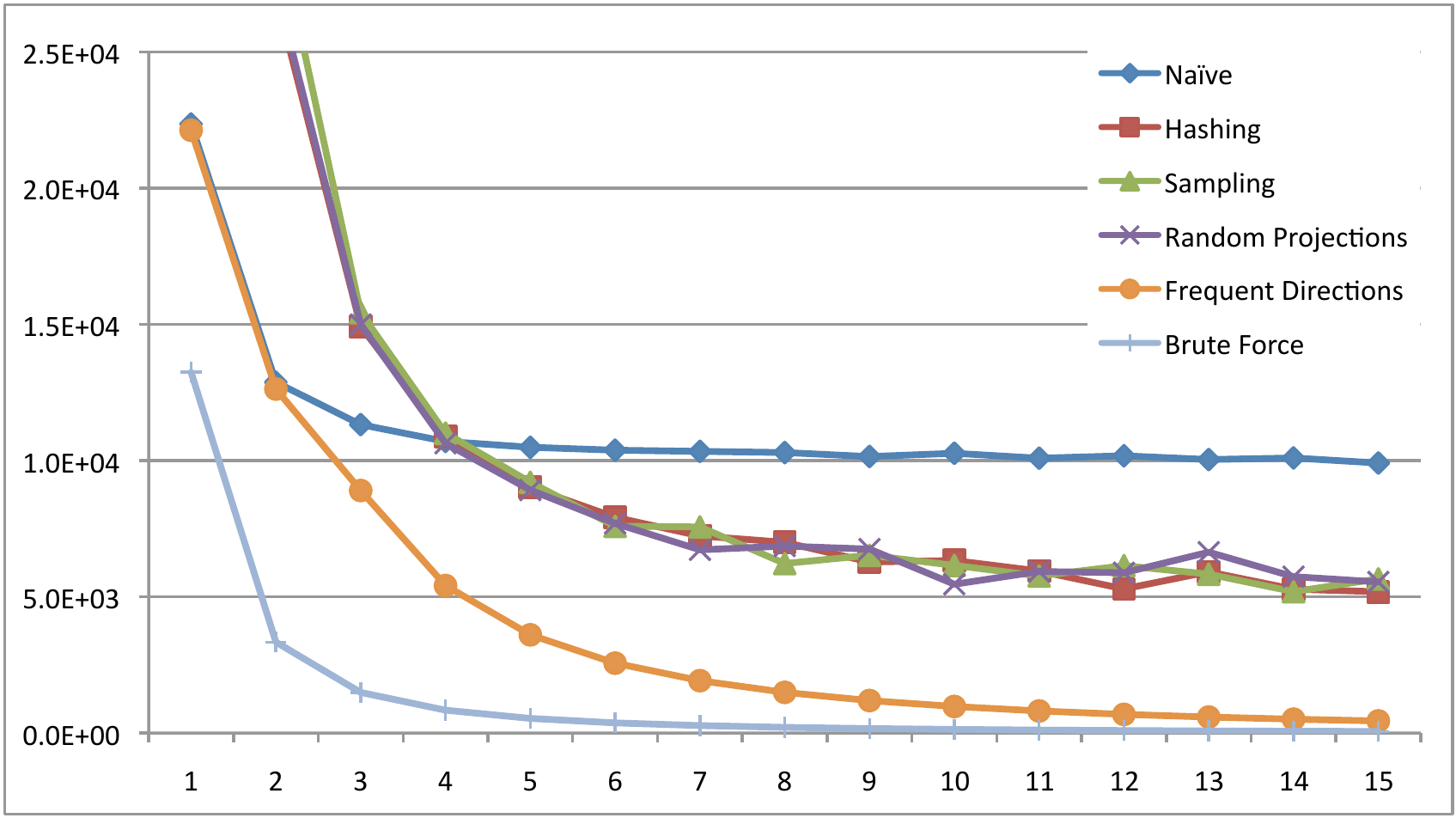}
\caption{Accuracy vs. signal to noise ratio. Here, all the sketches are of the same size $\ell=100$.
The sketches accuracies are measured as before and indicated on the $y$-axis. 
The $x$-axis gives the value of $\zeta$ the signal to noise ratio.
Here the signal dimension is $d=50$ and thus the brute force approach, 
which gives the best rank $100$ approximation of $A$, should theoretically recover $A$ almost exactly.
Note, however, that when $\zeta=1$ the best rank $100$ approximation of the matrix is quite poor.
That means that spectrum of $A$ is dominated by the noise. 
As $\zeta$ grows, the noise diminishes and the rows of $A$ become more concentrated around the $d$ dimensional signal row space.
Note that all the sketching techniques' results improve when the noise diminishes, as expected.  
}
\label{diff_two_vs_snr}
\end{center}
\end{figure}

\section{Discussion}
This paper draws upon a surprising similarity between two problems, the item frequency approximation problem and the matrix sketching problem.
It seems that, in general, solutions to the first can be modified to solve the second but incur 
an additional factor of $m$ in both running time and space requirement. 
This is true, for example, about sampling.
It is also the case for the memory footprint of \FI~which is $O(1/\eps)$ while for \FD~it is $O(m/\eps)$. 
But, the update time of \FI~is $O(1)$ and that of \FD~is $O(m/\eps)$. 
It is natural to seek a modified algorithm which exhibits an $O(m)$ update time.
Another question is whether more advanced algorithms for fining frequent items in streams could also be carried over.
A good candidate is the {\it Count Sketch} algorithm \cite{Charikar2002}. 
Alas, it depends on item hashing in a way which does not naturally translate to the matrix sketching domain.

\vspace{.5cm}
\noindent {\bf Acknowledgments:} The author truly thanks Petros Drineas, Jelani Nelson, Nir Ailon, Zohar Karnin, and Yoel Shkolnisky for very helpful discussions and pointers.
 
\pagebreak
\bibliographystyle{unsrt}
\bibliography{simpleMatrixSketching}

\end{document}